\pdfoutput=1
\documentclass[runningheads]{llncs}
\usepackage[utf8]{inputenc}
\usepackage[x11names,dvipsnames,table]{xcolor} 
\usepackage{mathtools}
\usepackage{amsmath,amsfonts,amssymb}
\usepackage{enumitem}
\usepackage{graphicx,subfigure}
\usepackage{wasysym}
\usepackage{algpseudocode,algorithm}
\usepackage{capt-of}
\usepackage{cleveref}
\usepackage{amsfonts}
\usepackage{array,multirow}

\usepackage{lineno}

\usepackage{flushend}
\usepackage{algpseudocode}
\usepackage{algorithm}

\usepackage[x11names,dvipsnames,table]{xcolor} 
\usepackage{mathtools}

\usepackage{mdframed}

\newmdenv[
backgroundcolor=gray!15,
topline=false,
bottomline=false,
skipabove=\topsep,
skipbelow=\topsep,
leftmargin=-5pt,
rightmargin=-5pt,
innertopmargin=3pt,
innerbottommargin=3pt
]{siderules}

\graphicspath{{figures/}}

\title{On Optimal $w$-gons in Convex Polygons}
\author{Vahideh Keikha}
	\institute{Institute of Computer Science, the Czech Academy of Sciences, Czech Republic, \email{keikha@cs.cas.cz}
	}
\begin{document}

\maketitle

\begin{abstract}
	Let $P$ be a set of $n$ points in $\mathbb{R}^2$. For a given positive integer $w<n$, our objective is to find a set $C \subset  P$ of points, such that $CH(P\setminus C)$ has the smallest number of vertices and $C$ has at most $n-w$ points. We discuss the $O(wn^3)$ time dynamic programming algorithm for monotone decomposable functions (MDF) introduced for finding a class of optimal convex $w$-gons, with vertices chosen from $P$, and improve it to $O(n^3 \log w)$ time, which gives an improvement to the existing algorithm for MDFs if their input is a convex polygon. 
\end{abstract}

\section{Introduction}
Let $P$ be a set of $n$ points in $\mathbb{R}^2$. 
We study a new variant of outlier detection for convex hull: Given $P$ and a positive integer $w<n$, find a subset $C \subset  P$ of points as the {\em outliers}, such that $CH(P\setminus C)$ has the smallest number of vertices, and $0\le |C|\le n-w$, where $|C|$ denotes the cardinality of $C$; see Figure~\ref{fig:def}. 
One of the motivations is reducing the input sample size to $w$ for support vector machine (SVM) classifiers~\cite{wang2013online}. 

\begin{definition}[Min-Size Convex Hull with Outliers (MinCH)]\label{def:problem}
	Let $P=\{p_1,\ldots,p_n\}$ be a set of $n$ points. For a positive integer $w$, we are interested in computing
	$\underset{\substack{C\subset P,\\|C| \le n-w}}{\arg \min}\; |CH(P\setminus C)|$.  Let $P^*_v$ denote $CH(P\setminus C)$.  
\end{definition}

\begin{figure}[h]
	\centering
	\includegraphics[scale=0.8]{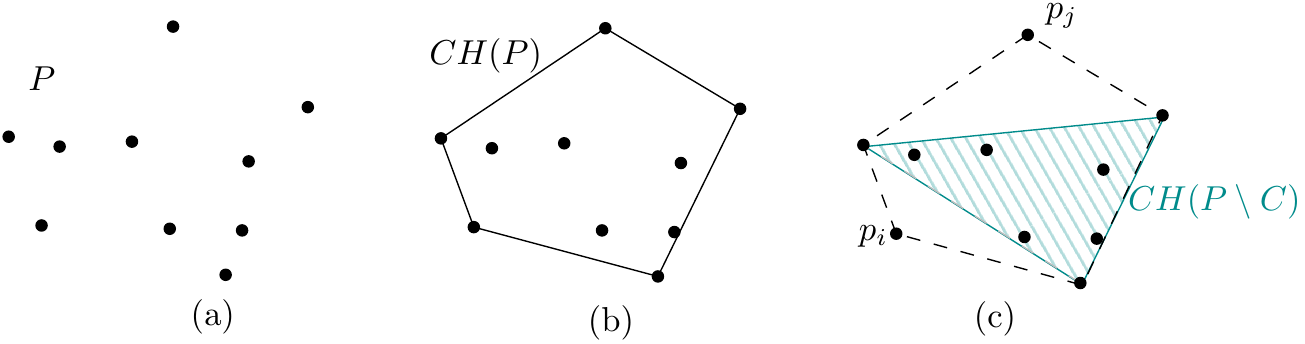}
	\caption{Problem definition: (a) A set $P$ of points with $n-w=2$;  (b) $CH(P)$; (c) $CH(P \setminus C)$ with $C=\{p_i,p_j\}$; the objective function is minimizing the number of vertices of $CH(P \setminus C)$. }
	\label{fig:def}
\end{figure}

A closely related problem is detecting at most $n-w$ points as outliers in $P$, such that after removing these outliers, the convex hull of the remaining points has the smallest area/perimeter. This problem can be solved in $O(n \log n)$ time for a constant number of outliers~\cite{atanassov2009algorithms}. 

Since we improve the algorithm of~\cite{eppstein1992finding} for computing a minimum area $w$-gon,  we also define this problem in the following.


\begin{definition}[Min-area $w$-gon~\cite{eppstein1992finding}]\label{def:problem}
	Let $P=\{p_1,\ldots,p_n\}$ be a set of $n$ points. For a positive integer $w<n$, we are interested in computing
	 a $w$-gon $P^*_a$ with vertices chosen from $P$, such that   $P^*_a$ has the smallest possible area among all choices. 
\end{definition}

\begin{figure}[t]
	\centering
	\includegraphics[scale=0.8]{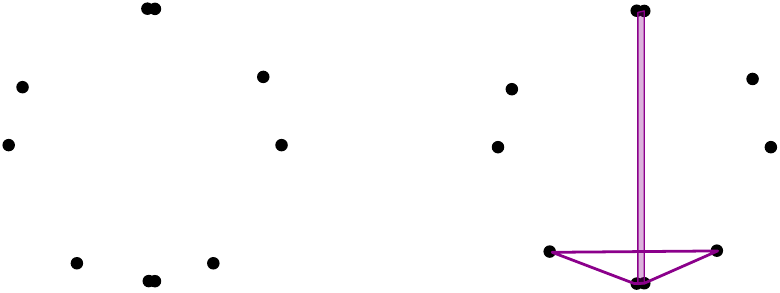}
	\caption{Problem definition on a set $P$ of 10 points: For $w=4$, the narrow rectangle is the solution to the Min-area, and the trapezoid is the solution to the Min-perimeter problem. }
	\label{fig:defarea}
\end{figure}

Min-perimeter $w$-gon would be defined analogously. See \Cref{fig:defarea} for an illustration.  In~\cite{eppstein1992finding}, a dynamic programming algorithm with $O(wn^3)$ time and $O(wn^2)$ space is given for Min-area $w$-gon, that carries over to solve any problem with an objective function in form of a {\em monotone decomposable 
	function (MDF)}, i.e., the objective function can be computed incrementally, but needs to change monotonicity increasing or decreasing:

\begin{definition}
	[Monotone Decomposable Function~\cite{eppstein1992finding}]\label{def:decomposable}
	
	A weight function $W$ is called monotone decomposable if and only if for any polygon $P$ and any index~$2<i<m$
	\begin{align*}
		W(P)= M (W(\langle p_1,\ldots,p_i\rangle),W(\langle p_1,p_i,p_{i+1}, \ldots, p_m\rangle),p_1,p_i)
	\end{align*}
	where $M$ can be computed in constant time, and $M$ is monotone in its first argument $W(\langle p_1,\ldots,p_i\rangle)$, which also means it is monotone in its second argument.
	
\end{definition}

So, computing the minimum/maximum perimeter $w$-gon, computing an optimal empty $w$-gon or computing an optimal polygon with $w$ as the number of vertices plus the interior points can be solved with the same technique.

Finding a possibly non-convex polygon of size $w$ with minimum/maximum area is known as {\em polygonalization}, and is NP-complete even in $\mathbb{R}^2$~\cite{fekete2000simple}. We note that it is not always possible  to find a strictly convex $w$-gon on any point set~\cite{mitchell1995counting}. 

If all the points of $P$ are on $CH(P)$, i.e. for points in convex position, computing the maximum area/perimeter $k$-gon takes $O(kn+n\log n)$ time~\cite{van2020maximum,boyce1985finding,aggarwal1987geometric}, and computing the minimum perimeter $k$-gon takes $O(n \log n+k^4 n)$ time~\cite{dobkin1983finding,aggarwal1991finding}.

The general idea of~\cite{eppstein1992finding} is to consider all optimal convex $(k-1)$-gons that can be extended by adding one more ear.
Presented algorithm uses several preprocessing on the points, e.g.,  
the clockwise ordering of all the vertices around each input point can be computed in $O(n^2)$ time~\cite{edelsbrunner1989topologically}.

 Another related problem is tabulating the number of convex $w$-gons for $w=3,\ldots,m$, that takes $O(mn^3)$ time~\cite{mitchell1995counting}.

We refer to~\cite{parhami2006introduction} for the standard models and definitions of distributed computations in the exclusive read exclusive write parallel random-access machine (EREW PRAM). 

\subsubsection*{Contribution}
 We improve the $O(wn^3)$ time dynamic programming algorithm of~\cite{eppstein1992finding} to $O(n^3 \log w)$, that improves the running time of all monotone decomposable weight functions if their input is a convex polygon (\Cref{sec:mdwf,sec:extension}). 
Our technique can also improve the running time of the dynamic programming algorithms in \cite{mitchell1995counting}. 
We discuss an easy extension of the algorithm for the class of  monotone decomposable weight functions (with convex input) to parallel settings (\Cref{sec:crew}).   \Cref{table:results} gives a summary of the new and known results.

\begin{table}[t]
	\centering
	\begin{tabular}{|p{2.5cm}|c|c|c|c|c|c|}
		\hline
		Measure &$\#$ Excluded pts. & Time & Space & Apprx. & Ref\\
		\hline\hline
		Min $\mathcal{A}/\mathcal{P}$  & $n-w$ & $O(wn^3)$&  $O(wn^2)$ & exact &\cite{eppstein1992finding} \\
		Min $\mathcal{A}/\mathcal{P}$ & $O(1)$ & $O(n \log n)$&  $O(n)$ & exact &\cite{atanassov2009algorithms} \\
		Min  $\#$ of interior points& $n-w$ & $O(wn^3)$&  $O(wn^2)$ & exact &\cite{eppstein1992finding} \\
		MDF & $n-w$ & $O(wn^3+G(n))$&  $O(wn^2)$ & exact &\cite{eppstein1992finding} \\
		MDF of convex polygons & $n-w$ & $O(n^3\log w)$&  $O(n^3\log w)$ & exact & Thm. \ref{thm:dp}\\
		Min	  $\#$ of CH vertices & $n-w$ & $O(n^3 \log w)$&  $O(n^3\log w)$ & exact &Thm. \ref{thm:dp} \\
		Min $\mathcal{A}/\mathcal{P}$ & $n-w$ & $O(n^3\log w)$ &  $O(n^3\log w)$ & exact &Thm. \ref{thm:dp} \\
		\hline
	\end{tabular}
	\caption{New and known results. $\mathcal{A}$ and  $\mathcal{P}$ stand for the area and the perimeter, respectively. $G(n)$ is the time complexity of computing the decomposable function on at most $O(n^3)$ triangles.}
	\label{table:results}
\end{table}

\section{The Dynamic Programming of~\cite{eppstein1992finding}}  \label{ap:eppalg}  
Since the original algorithm in~\cite{eppstein1992finding} is for minimizing the area of a $w$-gon as the weight function, we recall it with the original description. 
For two points $p_i$ and $p_j$, let $H_{i,j}$ denote the half-plane constructed to the right of the directed line passing through the line segment $\overrightarrow{p_ip_j}$. 

\textbf{Preprocess: Angular sort around points} For any point $p_i$, we first compute the sorted order of $P\setminus\{p_i\}$ in counter-clockwise (CCW) order around $p_i$ and denote it by $\phi(p_i)$.

The area of a triangle $p_ip_jp_l$ is defined as  $area(p_ip_jp_l)=|\overrightarrow{p_ip_l}\times \overrightarrow{p_ip_j} \sin \widehat {p_jp_ip_l}|$. 
The authors make a $4$-dimensional table $A$. Formally,  $A[p_i,p_j,p_l,m]$ denote the minimum area (at most) $m$-gon with $p_i$ as the bottommost vertex, $p_j$ as the next vertex of $p_i$ in  counter-clockwise order, and all the vertices lie on the same side of $p_jp_l$ as $p_i$. 
In the initialization, $A[*,*,*,2]=2$ ($*$ means any point in $P$), and they start the recursion from $A[*,*,*,3]$. 

In each iteration, the following cases should be considered. If $p_l \in H_{i,j}$, we cannot add a new vertex, and $A[p_i,p_j,p_l,m]=  A[p_i,p_l,pred(p_l),m]$, where $pred(p_l)$, is the last processed vertex in $\phi(p_j)$. 
If $p_l \notin H_{i,j}$, 
we may add one more vertex. So we have
$$\displaystyle A[p_i,p_j,p_l,m]=\min_{p_l \in P\setminus H_{i,j}} \big(A[p_i,p_l,p_j,m-1]+area(p_ip_jp_l),A[p_i,p_j,pred(p_l),m]\big),$$  
where $p_l$ is the first unprocessed predecessor neighbour of $p_j$ in $\phi(p_j)$; see Figure~\ref{fig:DP}. 
In the first term, $p_l$ is considered as a new vertex, while in the second term, omitting $p_l$ is more optimal.   
The total number of considered vertices in $P^*_a$ at most equals $w$, so we recurse from $3,\ldots,w$. A solution to $P^*_a$ is the minimum of $A[*,*,*,w]$. This algorithm is outlined in~\Cref{alg:DParea} and runs in $O(wn^3)$ time. 

\begin{figure}[t]
	\centering
	\includegraphics[scale=0.8]{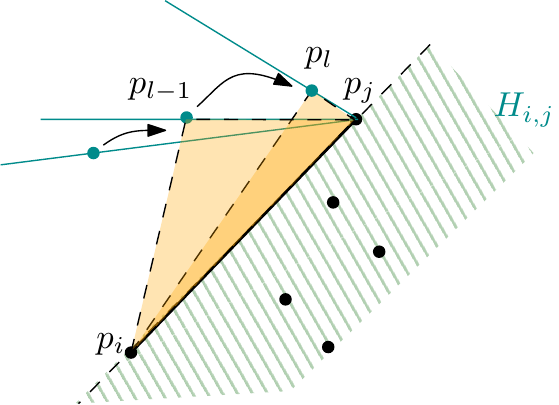}
	\caption{For any candidate $m$-gon with a fixed $p_i,p_j$, we check whether $p_l$ can be considered as a vertex (not lying on $H_{i,j}$), and also which of $p_l$ or $p_{l-1}$ determines the optimal choice.}
	\label{fig:DP}
\end{figure}

\begin{algorithm}[h]
	\caption{Min-area $w$-gon~\cite{eppstein1992finding}}
	\label{alg:DParea}
	\begin{algorithmic}[1]
		\Require{A set $P$ of points, $w>0$}
		\Ensure{A Min-area $w$-gon of $P$}
		\For{$p_i \in P$}
		\State{$A[p_i,p_l,p_j,2]=2,\quad \forall (p_j,p_l)\in P$}
		\For{$m=3,\ldots,w$}
		\For{$p_j\in P$, $p_l\in P$: $\overrightarrow{p_ip_l}\times \overrightarrow{p_ip_j} \leq0$} 
		\State{$v=\min_{p_l \in P\setminus H_{i,j}} A[p_i,p_l,p_j,m-1]+area(p_ip_jp_l)$}\label{line:cp}
		\EndFor
		\State{$\displaystyle A[p_i,p_j,p_l,m]=v$}
		\EndFor
		\EndFor\\
		\Return{$\min_{i,j,l} A[p_i,p_j,p_l,w]$}
	\end{algorithmic}
\end{algorithm}

\section{Exact Algorithms for MinCH} \label{sec:MainDP}
As $P^*_v$ can be decomposed into a triangle and a $(w-1)$-gon with the optimal number of vertices, we can find $P^*_v$ recursively. We use this idea combined with a divide-and-conquer technique to improve the running time of~\cite{eppstein1992finding}. We first prove that MinCH is an MDF. 

\begin{lemma} MinCH is a monotone decomposable function.
	\end{lemma}
\begin{proof}
	First observe that one can
cut $P^*_v$ by any chord $p_1p_i$, where $p_1,p_i \in P^*_v$, and summing over the number of vertices 
of the induced sub-polygons minus 2 (using the general position assumption in~\cite{eppstein1992finding} that no three points are collinear) determines the number of vertices of $P^*_v$. Hence, $$M(|CH(p_1,\ldots,p_i)|,|CH(p_1,p_i,p_{i+1},\ldots,p_w)|,p_1,p_i)=$$
$$|CH(p_1,\ldots,p_i)|+|CH(p_1,p_i,p_{i+1},\ldots,p_w)|-2=w.$$
\end{proof}

\subsection{Adjusting the Dynamic Programming of~\cite{eppstein1992finding} to MinCH} \label{ap:minchrote}
For adjusting~\Cref{alg:DParea} for minimizing the number of vertices of the convex hull, the recursion would be $A[p_i,p_j,p_l,m]= \underset {p_l \notin H_{i,j}} {\min}\big( A[p_i,p_l,p_j,m-1]+1,A[p_i,p_j,pred(p_l),m]\big),$
where $p_l$ is the first unprocessed predecessor neighbour of $p_j$ in $\phi(p_j)$. 
In the first term, $p_l$ is considered as a new vertex, while in the second term, omitting $p_l$ is more optimal.



\begin{algorithm}[t]
	\caption{Improved Min-area $w$-gon }
	\label{alg:betterminch}
	\begin{algorithmic}[1]
		\Require{A set $P$ of points, $w>0$}
		\Ensure{A Min-area $w$-gon of $P$}
		\State{$\phi(p_i)=$the sorted list of $p_j, j\ne i$ in CCW  direction around $p_i, \forall p_i\in P$}
		\State{$Cost[i][j][l]=\infty,~End[i][j]= Nil, \forall i\ne j$ ($Nil$: nothing in list)}
		\State{$Cost[i][i][l]=0,~End[i][i]= i$}
		\For{$p_i \in P$}
		\For{$w=2^t, t=0,\ldots,\log_2 w$}
		\For{$p_j$ in the order of $\phi(p_i),j\ne i$}
		\For{$p_r$ in the order of $\phi(p_j),r\ne i,j$ and $p_ip_jp_r$ is CCW}
		\State{$l=End[i][r]$}
		\If{$Cost[i][j][l]> Cost[i][j][r]+ Cost[i][r][l]$} \label{line:costdef}
		\State{$Cost[i][j][l]= Cost[i][j][r]+ Cost[i][r][l]$} \label{line:computeweight}
		\State{$End[i][j]=l$}
		\EndIf
		\EndFor
		\EndFor
		\EndFor
		\EndFor \\
		\Return{$\min_{i,j,l} Cost[i][j][l]$}
	\end{algorithmic}
\end{algorithm}

\subsection{An $O(n^3 \log w)$ Time Algorithm} \label{sec:mdwf}

As we give an improvement for the dynamic programming algorithm in~\cite{eppstein1992finding}, and the original algorithm in~\cite{eppstein1992finding} is formulated to find a Min-area $w$-gon  (but it is shown that it works for the broader class MDF), we first adopt our explanation based on finding a $w$-gon of smallest area, and then we discuss the generalizations.

Our idea is that we do not recurs on $w$, and instead, we design a divide-and-conquer algorithm, at which we merge two convex polygons of size $2^{t}$ vertices in each iteration, for $t=0,\ldots,\log w$. Then we merge the constructed sub-polygons.
Each sub-polygon is constructed based on the dynamic programming algorithm of~\cite{eppstein1992finding}.

For ease of exposition, let $w=2^t$ for some $t>0$. 
We define a 3-dimensional array $Cost[i][j][r]$, that stores the optimal objective function on a convex polygon with the lowest vertex $p_i$, and with $p_j$ as the next vertex in CCW direction, and $p_r$ is the first unprocessed predecessor neighbour of $p_j$ in $\phi(p_j)$.

In Line~\ref{line:computeweight} of \Cref{alg:betterminch}, $Cost[i][j][l]$ achieves the optimal objective function resulted by merging two smaller polygons stored in $Cost[i][j][r]$ and $Cost[i][r][l]$, each with half the number of vertices of $Cost[i][j][l]$.  
We kept the index $l$ at another array $End[i][j]$, which we would know the last vertex of the constructed convex chain on $p_i$ and $p_r$ so far is $p_l$. 
We remove this assumption of~\cite{eppstein1992finding} that all the vertices lie on the same side of the supporting line of $p_jp_l$ as $p_i$, that is because we merge two convex polygons of $t\ge 3$ vertices instead of merging a polygon with a triangle. Keeping the convexity constraints at the vertex ($p_r$) we perform the merging prevents making non-feasible solutions.  The initialization rules of the tables are straightforward. 

In Min-area $w$-gon problem, we keep the area of the constructed sub-polygons in $Cost[i][j][l]$, and return $\min_{i,j,l} Cost[i][j][l]$ at the end. 

\subsubsection{Correctness}  Observe that we consider all possible configurations for which three vertices $p_i,p_j$ and $p_l$ define an optimal convex polygon, 
where $p_i$ is the lowest vertex, $p_j$ is the next vertex of $p_i$ in CCW and $p_l$ is the vertex immediately before $p_i$ in CCW, and the constructed angle at $\widehat{ p_ip_jp_r}$ is CCW (i.e., the inner angle at $p_r$ is convex), where $p_r$ is the vertex we concatenate two small polygons and make a polygon with $p_ip_r$ as a chord. 

In Algorithm~\ref{alg:betterminch}, throughout $O(\log w)$ steps, 
we compute the optimal $w$-gon (among all candidates) by concatenating two optimal convex polygons each of $2^t$ vertices for $t=1,\ldots,\log w$.  Optimality of the concatenated sub-polygons and holding the convexity constraint at the concatenating vertex $p_r$ guarantee the correctness and the optimality of the reported solution. (Note that the reported polygon has indeed $2w$ vertices, but the adjustment is straightforward.)

\subsubsection{Extension to MinCH} \label{sec:minch_ex}
In MinCH, $Cost[i][j][l]$ keeps the number of vertices, and we return $\min_{i,j,l} Cost[i][j][l]$. The initialization rules of the tables are straightforward. 

\subsubsection*{Extension to Other Monotone Decomposable Functions} \label{sec:extension}
For computing Min-perimeter $w$-gon, we keep the perimeter in $Cost[i][j][l]$, but in line \ref{line:computeweight}, we also remove the weight of the edge $p_ip_r$ multiplied by 2 from $Cost[i][j][l]$. We note that one may need to alter the condition of Line \ref{line:costdef} of Algorithm~\ref{alg:betterminch} proportional to the minimization or maximization criterion. 

We add the convexity condition to \Cref{def:decomposable} to make the 
extension to the other weight functions (if the weight function is convex decomposable):

\begin{definition}
	[Convex Decomposable Function]\label{def:cdecomposable}
	
	A weight function $W$ is called convex decomposable if and only if for any convex polygon $P=\langle p_1,\ldots,p_m\rangle$ and any index $2<i<m$
	\begin{align*}
		W(P)= M (W(\langle p_1,\ldots,p_i\rangle),W(\langle p_1,p_i,p_{i+1}, \ldots, p_m, p_1\rangle))
	\end{align*}
	where $M$ can be computed in constant time, and $M$ is semigroup in its first argument $W(\langle p_1,\ldots,p_i\rangle)$, which also means it is semigroup in its second argument.
	
\end{definition}

\begin{theorem} \label{thm:dp} Let $P$ be a set of $n$ points, and $0 \le w<n$ be an integer. An optimal solution to the MinCH, Min-area $w$-gon, Min-perimeter $w$-gon and any convex decomposable function can be computed in $O(n^3\log w)$ time and $O(n^3 \log w)$ space. 
\end{theorem}
\begin{proof} The correctness of the algorithm is already discussed. 
	The running time follows from the fact that we do not recurse on $w$ to perform one vertex insertion per iteration; we merge two convex polygons of $2^t$ vertices for $t=0,\ldots,\log w$. 
	As our computations are not in place, obviously $O(n^3 \log w)$ space is required to run the algorithm.   
\end{proof}

We finally note that

\begin{itemize}
	
	\item the weight function and the objective function are not required to be necessarily the same.

	\item
	another generalization is the case where the given $w$ is a bound on a continues measure, e.g., a bound on the area or the perimeter. Then the objective function is to find a convex hull with the smallest number of vertices so that the area or the perimeter is at most $w$. The set of candidate solutions is still the ones computed by~\Cref{alg:betterminch}.
\end{itemize}

\section{Extension to Distributed Environments}\label{sec:crew}\label{sec:MPC}
In this section, we show that any problem with a convex decomposable weight function can be easily transformed to work in parallel settings. Indeed as the particles of the solution are convex, the order of concatenation is not important. 
\Cref{alg:betterminch} in EREW PRAM, using $O(n^2)$ space and $O(n^3)$ processors, works by evaluating any triple $p_ip_jp_k,i\ne j\ne k$ in one processor, which takes $O(\log w)$ time. Merging the sub-problems takes $O(\log n)$ time, using parallel prefix min~\cite{parhami2006introduction}. Computation of $\phi(p_i)$, for $i=1,\ldots,w$ also takes $O(\log n)$ time in this model.
So, the overall time is $O(\log w \log n)$ and the total work is $O(n^3 \log w \log n)$.

The parallel version of \Cref{alg:DParea} (restated from~\cite{eppstein1992finding}, as the algorithm is sequential itself) is analysed similarly, where $w$ values are checked instead of $\log w$ values; resulting in $O(w\log n)$ time and $O(n^3 w \log n)$ work.

What we discussed in this section is not quite work optimal, but any improvement on this gives an improvement to the work of a broad class of problems (belonging to MDFs) in parallel settings. 


\section{Discussion} 
Further improvements to the running time of the problems in the class MDF  or providing a lower bound for the problems in this class remained open. 

\paragraph*{Acknowledgement} The author would like to thank Sepideh Aghamolaei for all the fruitful discussions on the results of this paper.  V.Keikha is supported by
the Czech Science Foundation, grant number GJ19-06792Y, and with institutional support RVO:67985807.

\bibliographystyle{abbrv}

\bibliography{refs}

\end{document}